\documentclass[12pt,a4paper]{article}
\usepackage{jheppub}
\usepackage{amsthm,mathtools,bm,booktabs,array,microtype}
\usepackage[T1]{fontenc}
\usepackage{lmodern}
\usepackage{xcolor}

\newtheorem{theorem}{Theorem}[section]
\newtheorem{proposition}[theorem]{Proposition}
\newtheorem{lemma}[theorem]{Lemma}
\newtheorem{corollary}[theorem]{Corollary}
\theoremstyle{remark}\newtheorem{remark}[theorem]{Remark}

\newcommand{\E}{\mathbb E}
\newcommand{\Var}{\mathrm{Var}}

\newcommand{\He}{\mathrm{He}}
\newcommand{\dd}{\mathrm d}
\newcommand{\phiG}{\varphi}
\newcommand{\kfour}{\kappa_4}
\newcommand{\norm}[1]{\left\lVert #1\right\rVert}

\title{Noise-Induced Restoration of Entropy Concavity Near Gaussian Equilibrium}
\author{Fayez Abu-Ajamieh$^1$,}
\affiliation{$^1$Formerly: The Indian Institute of Science, Bangalore, India}
\emailAdd{fayezabuajamieh@gmail.com}

\abstract{We study a quantitative form of noise-induced restoration of entropy concavity for one-dimensional nonequilibrium states relaxing toward a Gaussian thermal equilibrium. A partially thermalized state is modeled by $X_q=\sqrt q\,X+\sqrt{1-q}\,G$, where $G$ is Gaussian equilibrium noise and $q$ is the retained signal fraction. For two independently prepared copies we ask when the entropy of the mixed state $h(\sqrt{1-t}\,X_q+\sqrt t\,Y_q)$ is concave in the mixing parameter. The problem is motivated by entropy geometry along Ornstein--Uhlenbeck/Fokker--Planck relaxation, where the relative score is the nonequilibrium thermodynamic force and relative Fisher information is entropy-production rate. We derive an exact symmetric endpoint-curvature formula, separating a nonnegative Ornstein--Uhlenbeck restoring term from a fourth-cumulant correction, and extend it to nonsmooth generalized-Gaussian states $e^{-a|x|^p}$ for $p>3/2$. For half thermalization of a symmetric log-concave input, posterior-variance identities give a quantitative outer-region theorem confining any endpoint obstruction to two explicitly bounded near-equilibrium sectors. We also give spectral and MMSE reformulations of the auxiliary inequality $A\ge2\delta$, prove it in an explicit nonlinear neighborhood of Gaussian equilibrium, and show how it sharpens the endpoint bounds. The main all-mixing result concerns the leading symmetric non-Gaussian Hermite mode: for $p_a=\phiG(1+a\He_4)$, $0<a\le10^{-2}$, we prove strict entropy concavity for every $0<t<1$ and every $0<q\le5/8$, hence throughout the currently unresolved symmetric range $0<q\le4/7$ for this explicit nonequilibrium family. Numerical searches over broader log-concave families are reported only as supporting evidence.}
\arxivnumber{}

\begin{document}
\maketitle
\flushbottom

\section{Introduction}

Entropy production and relaxation to equilibrium are central themes in nonequilibrium statistical mechanics. In the simplest one-dimensional model, an overdamped degree of freedom in a quadratic confining potential evolves by the Ornstein--Uhlenbeck/Fokker--Planck equation
\begin{equation}
\partial_\tau p_\tau(x)=\partial_x(xp_\tau(x))+\partial_x^2p_\tau(x),
\label{eq:fp-intro}
\end{equation}
whose stationary state is the Gaussian equilibrium density $\phiG(x)=(2\pi)^{-1/2}e^{-x^2/2}$. Along this flow, relative entropy and relative Fisher information quantify free-energy dissipation and its rate of decay; in physics language the relative score
\begin{equation}
s(x)=\partial_x\log\frac{p(x)}{\phiG(x)}
\end{equation}
acts as the thermodynamic force driving relaxation \cite{Risken1989,Seifert2012,Villani2009}. These quantities also appear naturally in information-theoretic Gaussian channels and in de Bruijn-type entropy identities \cite{GuoShamaiVerdu2005,ZhangAnantharamGeng2018}.

This paper studies a related but different entropy-geometry question. Start from a symmetric nonequilibrium state $X$ with variance one and partially thermalize it by adding Gaussian bath noise,
\begin{equation}
X_q=\sqrt q\,X+\sqrt{1-q}\,G,
\label{eq:thermalized-intro}
\end{equation}
where $0\le q\le1$ measures the fraction of the variance still carried by the initial state. Thus $q=1$ is the unthermalized state and $q=0$ is equilibrium. Given two independently prepared copies $X_q$ and $Y_q$, we ask whether the entropy of their linear mixture,
\begin{equation}
F_q(t)=h\!\left(\sqrt{1-t}\,X_q+\sqrt t\,Y_q\right),\qquad 0<t<1,
\label{eq:physics-Fq}
\end{equation}
is concave in the mixing parameter $t$. Concavity expresses a stability property of entropy under coherent interpolation between two independent preparations.

The same question also sits inside a recent mathematical problem. Ball, Nayar and Tkocz conjectured entropy concavity for weighted sums of independent log-concave random variables \cite{BallNayarTkocz2016}. Luo recently disproved the unrestricted conjecture through an asymmetric mechanism, while identifying a remaining symmetric low-signal regime \cite{Luo2026}. Our viewpoint is that Gaussian noise should restore part of the lost entropy geometry, at least near equilibrium and for physically natural symmetric perturbations. We therefore phrase the results as noise-induced restoration theorems, while keeping the precise probability-theoretic hypotheses explicit.

The most physics-transparent perturbative family is obtained by expanding around equilibrium in Hermite modes. Symmetry removes odd modes, and the variance normalization removes the $\He_2$ mode. The first nontrivial even departure from a Maxwellian/Gaussian equilibrium is therefore
\begin{equation}
p_a(x)=\phiG(x)\bigl(1+a\He_4(x)\bigr),
\label{eq:h4-intro}
\end{equation}
which is the analogue of the leading kurtosis/Sonine correction familiar from kinetic theory \cite{Grad1949,Cercignani1988}. A main theorem below proves strict entropy concavity for this family for every mixing fraction $0<t<1$ and every signal fraction $0<q\le5/8$. This includes the entire symmetric range $0<q\le4/7$ left open by Luo for the general problem, but only for the explicit leading-mode family \eqref{eq:h4-intro}.

The rest of the paper develops the analytic structure behind this restoration phenomenon. First we derive an exact symmetric endpoint-curvature identity. It splits the curvature into a positive Ornstein--Uhlenbeck restoring term and a genuinely non-Gaussian fourth-cumulant correction. Second, we prove endpoint concavity for the whole generalized-Gaussian family $e^{-a|x|^p}$, including the nonsmooth range $3/2<p<3$ via a weak-regularity argument. Third, for half thermalization $U=(X+G)/\sqrt2$ of a symmetric log-concave input, posterior-variance identities give quantitative exclusion regions for endpoint obstructions. Fourth, an exact spectral/MMSE reformulation identifies the auxiliary inequality whose proof would further shrink the remaining near-equilibrium sector. Numerical searches are included only as validation and guidance, not as theorem-level evidence.

The theorem-level results are summarized in table~\ref{tab:main-results}. The full symmetric log-concave conjecture is not proved here.

\begin{table}[t]
\centering
\begin{tabular}{p{0.27\textwidth}p{0.62\textwidth}}
\toprule
Result & Rigorous content \\
\midrule
Endpoint identity & Exact formula for $F''(0^+)$ separating OU restoration from a fourth-cumulant correction.\\
Generalized Gaussians & Endpoint concavity for $f_p\propto e^{-a|x|^p}$, $p>3/2$, using a weak-regularity limit.\\
Half-thermalized log-concave states & Quantitative outer-region exclusion theorem for endpoint obstructions.\\
Spectral/MMSE reformulation & Exact representation of $A-2\delta$ and conditional improvement of the outer thresholds.\\
Gaussian neighborhood & Proof of $A\ge2\delta$ for $\|p/\phiG-1\|_\infty\le0.1884695274\ldots$.\\
Leading Hermite mode & Strict all-$t$ entropy concavity for $\phiG(1+a\He_4)$, $0<a\le10^{-2}$, $0<q\le5/8$.\\
\bottomrule
\end{tabular}
\caption{Theorem-level results established in this work. Exploratory searches are deliberately excluded from the theorem list.}
\label{tab:main-results}
\end{table}

\section{Thermalization setup and notation}

The Gaussian smoothing in \eqref{eq:thermalized-intro} is exactly the time-$\tau$ Ornstein--Uhlenbeck semigroup after setting $q=e^{-2\tau}$:
\begin{equation}
p_q=\mathrm{Law}\left(e^{-\tau}X+\sqrt{1-e^{-2\tau}}\,G\right).
\end{equation}
Thus small $q$ means long thermal contact with the Gaussian bath. The Hermite modes diagonalize this relaxation: a mode of degree $n$ is damped by $q^{n/2}$. In particular, the leading symmetric variance-preserving mode $\He_4$ is damped by $q^2$, a fact used exactly in section~\ref{sec:h4}.

All logarithms are natural. Let $f$ be an even probability density on $\mathbb R$ with mean zero and variance one. Its score is
\begin{equation}
\rho(x)=\frac{f'(x)}{f(x)}=(\log f)'(x).
\label{eq:score}
\end{equation}
When $f=e^{-V}/Z$, one has $\rho=-V'$. The fourth cumulant is
\begin{equation}
\kfour=\E [X^4]-3.
\end{equation}
We write
\begin{equation}
\phiG(x)=\frac{e^{-x^2/2}}{\sqrt{2\pi}}
\end{equation}
for the standard Gaussian density. The relative score of a variance-one density $p=r\phiG$ is
\begin{equation}
s=(\log r)'=(\log p)'+x.
\label{eq:relative-score}
\end{equation}
Integration by parts gives
\begin{equation}
\E[X\rho(X)]=-1,
\qquad
\E[Xs(X)]=0.
\label{eq:score-orthogonality}
\end{equation}
The second identity is an exact consequence of the variance normalization and will remove the first odd Gaussian mode in section~\ref{sec:local}.

\section{Exact symmetric endpoint curvature}
\label{sec:endpoint}

\begin{theorem}[Smooth endpoint identity]
\label{thm:smooth-endpoint}
Let $X,Y$ be independent copies of a smooth strictly positive even variance-one random variable with density $f$ and score $\rho$. Assume sufficient decay for the integrations by parts below. Then
\begin{equation}
\boxed{
F''(0^+)
=-\frac12\E[(\rho'+1)^2]
-\frac{\kfour}{12}
\left(\E[(\rho')^2]-\frac13\E[\rho^4]\right).}
\label{eq:endpoint-main}
\end{equation}
Equivalently, for $f=e^{-V}/Z$,
\begin{equation}
F''(0^+)
=-\frac12\E[(V''-1)^2]
-\frac{\kfour}{12}
\left(\E[(V'')^2]-\frac13\E[(V')^4]\right).
\label{eq:endpoint-V}
\end{equation}
\end{theorem}

\paragraph{Derivation.}
Set $t=u^2$ and write $Q(u)=F(u^2)$. Symmetry of $Y$ makes $Q$ even. For a smooth density the interpolated density has
\begin{equation}
p_u=f+c_2u^2+c_4u^4+O(u^6),
\end{equation}
where
\begin{align}
c_2&=\frac12(f+xf'+f''),\\
2c_4&=\frac34f+\frac54xf'+\frac14x^2f''+\frac32f''+\frac12xf'''+\frac{\mu_4}{12}f'''',
\end{align}
and $\mu_4=\E [X^4]$. Differentiating the entropy gives
\begin{equation}
F''(0^+)=-\int (2c_4)\log f\,\dd x-\int \frac{c_2^2}{f}\,\dd x.
\label{eq:entropy-second}
\end{equation}
The $\mu_4=3$ part is the normalized Ornstein--Uhlenbeck interpolation and equals
\begin{equation}
-\frac12\E[(\rho'+1)^2].
\end{equation}
The remaining fourth-cumulant term is
\begin{equation}
-\frac{\kfour}{12}\int f''''\log f\,\dd x.
\end{equation}
Two integrations by parts and
\begin{equation}
\E[\rho^2\rho']=-\frac13\E[\rho^4]
\label{eq:score-id}
\end{equation}
give
\begin{equation}
\int f''''\log f
=\E[(\rho')^2]-\frac13\E[\rho^4],
\end{equation}
which proves \eqref{eq:endpoint-main}.

The smoothness hypothesis above is convenient but stronger than the integrability actually needed by the final formula. The following lemma is included because it is essential for the generalized-Gaussian application.

\begin{lemma}[Weak-regularity extension for convex-power potentials]
\label{lem:weak-gg}
Let
\begin{equation}
f_p(x)=Z_p^{-1}e^{-a|x|^p},\qquad p>\frac32,
\end{equation}
with $a$ chosen so that $\E [X^2]=1$. Then $F_p''(0^+)$ exists and is given by \eqref{eq:endpoint-V}, with $V'=ap\,\mathrm{sgn}(x)|x|^{p-1}$ and $V''=ap(p-1)|x|^{p-2}$ almost everywhere.
\end{lemma}

\begin{proof}
Use the smooth even convex approximation
\begin{equation}
V_{p,\varepsilon}(x)=a_\varepsilon(x^2+\varepsilon^2)^{p/2},
\end{equation}
where $a_\varepsilon$ is chosen so that the associated density has variance one. For $p>1$ these potentials are convex and smooth. As $\varepsilon\downarrow0$, $a_\varepsilon\to a$, the densities converge in $L^1$ with all polynomial moments needed below, and $V'_{p,\varepsilon}\to V'$ pointwise away from the origin. Moreover
\begin{equation}
|V''_{p,\varepsilon}(x)|^2
\le C_p(x^2+\varepsilon^2)^{p-2}
\end{equation}
on compact sets. The right-hand side is uniformly locally integrable exactly when $p>3/2$. Exponential tails give uniform integrability at infinity. Thus
\begin{equation}
\E[(V''_{p,\varepsilon})^2]\to\E[(V'')^2],
\qquad
\E[(V'_{p,\varepsilon})^4]\to\E[(V')^4],
\qquad
\kappa_{4,\varepsilon}\to\kappa_4.
\end{equation}
Applying theorem~\ref{thm:smooth-endpoint} at fixed $\varepsilon$ and passing to the limit gives the endpoint quadratic coefficient. Appendix~\ref{app:weak} spells out the associated second-order entropy difference quotient. Importantly, no $O(u^6)$ remainder is asserted for $f_p$ itself: for non-even-integer $p$ the density need not be $C^4$ at the origin and noninteger powers of $t$ can appear beyond the $t^2$ term. Only the existence and value of $F''(0^+)$ are used here.
\end{proof}

Define
\begin{equation}
A:=\E[(\rho'+1)^2],
\qquad
B:=\E[(\rho')^2]-\frac13\E[\rho^4].
\label{eq:AB}
\end{equation}
Then endpoint concavity is equivalent to
\begin{equation}
\boxed{\kfour B+6A\ge0.}
\label{eq:scalar-target}
\end{equation}

\section{Generalized Gaussians and two false shortcuts}

Consider the standardized generalized-Gaussian family
\begin{equation}
f_p(x)=Z_p^{-1}e^{-a|x|^p},\qquad p>\frac32.
\end{equation}
The moment formula is
\begin{equation}
\E|X|^r=a^{-r/p}\frac{\Gamma((r+1)/p)}{\Gamma(1/p)},
\qquad
a=\left[\frac{\Gamma(3/p)}{\Gamma(1/p)}\right]^{p/2}.
\label{eq:ggmom}
\end{equation}

\begin{proposition}[Generalized-Gaussian endpoint concavity]
\label{prop:gg}
For every $p>3/2$,
\begin{equation}
F_p''(0^+)\le0,
\end{equation}
with equality only for the Gaussian member $p=2$.
\end{proposition}

\begin{proof}
Lemma~\ref{lem:weak-gg} justifies the endpoint formula throughout the stated range. Direct evaluation using \eqref{eq:ggmom} gives
\begin{equation}
B=C_p(2-p),
\qquad
C_p=a^{4/p}p^2(p-1)\frac{\Gamma(2-3/p)}{\Gamma(1/p)}>0.
\end{equation}
The standardized fourth moment is
\begin{equation}
\mu_4(p)=\frac{\Gamma(5/p)\Gamma(1/p)}{\Gamma(3/p)^2}.
\end{equation}
It is $>3$ for $p<2$, $=3$ at $p=2$, and $<3$ for $p>2$. Therefore $\mathrm{sign}(\kfour)=\mathrm{sign}(B)$ and $\kfour B\ge0$. Equation~\eqref{eq:endpoint-main} gives the result.
\end{proof}

\begin{figure}[t]
\centering
\includegraphics[width=.72\textwidth]{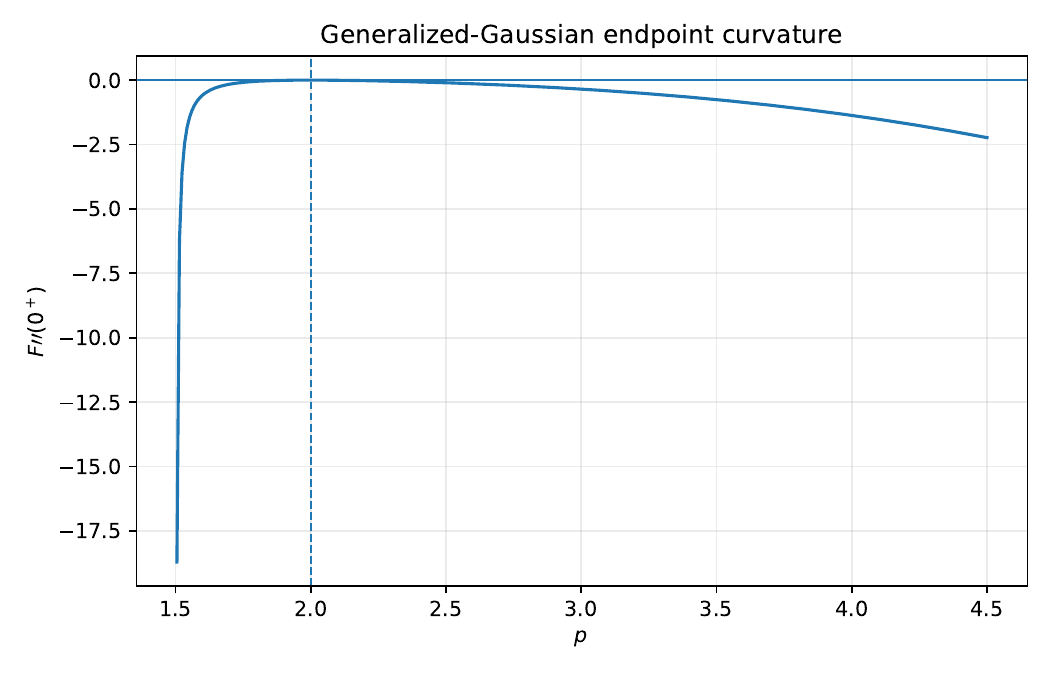}
\caption{Exact endpoint curvature for the standardized generalized-Gaussian family. The plotted formula is valid for all $p>3/2$ by lemma~\ref{lem:weak-gg}; the Gaussian member $p=2$ saturates the result.}
\label{fig:gg}
\end{figure}

The master identity also rules out two tempting proof strategies. First, a pointwise posterior-covariance inequality can fail even when its density-weighted average has the correct sign. Second, the sign rule $\kfour B\ge0$ is not universal. A stable half-Gaussian-smoothed piecewise-convex example gives
\begin{equation}
\kfour=7.718249835618\times10^{-4},\qquad
B=-3.206474139903\times10^{-4},
\end{equation}
while
\begin{equation}
A=3.925033377823\times10^{-3},\qquad
F''(0^+)\simeq-1.9624961\times10^{-3}.
\end{equation}
This is a numerical diagnostic, not a theorem, but it shows that the target is genuinely \eqref{eq:scalar-target} rather than a stronger sign statement.

\section{Half-Gaussian smoothing and the outer-region theorem}
\label{sec:outer}

Let $X$ be symmetric log-concave with mean zero and variance one and put
\begin{equation}
U=\frac{X+G}{\sqrt2}.
\label{eq:halfnoise}
\end{equation}
Let the density of $U$ be $f=e^{-V}/Z$ and define
\begin{equation}
g=V''=-\rho',\qquad h=g-1.
\end{equation}
If
\begin{equation}
v(z)=\Var(X\mid X+G=z),
\end{equation}
then Gaussian-channel differentiation identities \cite{DytsoPoorShamai2021} give
\begin{equation}
g(U)=2(1-v(X+G)).
\label{eq:posterior-g}
\end{equation}
The posterior is $1$-strongly log-concave. Brascamp--Lieb therefore gives $0\le v\le1$ \cite{BrascampLieb1976,SaumardWellner2014}, so
\begin{equation}
0\le g\le2,\qquad |h|\le1.
\label{eq:g-bounds}
\end{equation}
Since $A=\E [h^2]$, set
\begin{equation}
I=\E[\rho^2]=\E [g],
\qquad
\delta=I-1=\E [h].
\end{equation}
Then
\begin{equation}
\delta^2\le A.
\label{eq:deltabound}
\end{equation}
If $K=\E [g^2]$, then
\begin{equation}
K=1+A+2\delta\le(1+\sqrt A)^2.
\label{eq:Kbound}
\end{equation}

The correction $B$ has a centered form. From \eqref{eq:score-id},
\begin{equation}
\E[g\rho^2]=\frac13\E[\rho^4].
\end{equation}
Set $D=g-\rho^2$. Then $\E [D]=0$ and
\begin{equation}
B=\E[hD],
\qquad
\E [D^2]=2K-B.
\end{equation}
Therefore
\begin{equation}
\boxed{B^2\le A(2K-B).}
\label{eq:centeredCS}
\end{equation}

The final input is a sharp fourth-moment window. For a symmetric variance-one log-concave random variable,
\begin{equation}
\frac95\le\E [X^4]\le6.
\label{eq:sharp-fourth}
\end{equation}
These are sharp log-concave moment bounds; the extremal description follows from the log-concave moment problem of Eskenazis, Nayar and Tkocz \cite{EskenazisNayarTkocz2018Moments}. Since fourth cumulants add under independent convolution and scale homogeneously,
\begin{equation}
\kappa_4(U)=\frac14\kappa_4(X),
\end{equation}
so
\begin{equation}
\boxed{-\frac3{10}\le\kfour\le\frac34.}
\label{eq:kurt-range}
\end{equation}

\begin{theorem}[Quantitative outer-region endpoint theorem]
\label{thm:outer}
For the half-Gaussian-smoothed density above, $F''(0^+)\le0$ in each of the following cases:
\begin{enumerate}
\item $\kfour B\ge0$;
\item $\kfour>0$, $B<0$, and
\begin{equation}
A\ge A_+:=\frac{29+4\sqrt7}{729}=0.054297675232\ldots;
\end{equation}
\item $\kfour<0$, $B>0$, and
\begin{equation}
A\ge A_-:=\left(\frac{1+\sqrt{210}}{209}\right)^2=0.005493984879\ldots.
\end{equation}
\end{enumerate}
Hence any endpoint counterexample must lie in one of two explicitly near-Gaussian sign-mismatch sectors.
\end{theorem}

\begin{proof}
If $\kfour B\ge0$, equation~\eqref{eq:scalar-target} is immediate. Suppose $\kfour>0$ and $B=-b<0$. Equation~\eqref{eq:centeredCS} gives
\begin{equation}
b\le\frac{A+\sqrt{A^2+8AK}}2.
\end{equation}
With $x=\sqrt A$ and \eqref{eq:Kbound},
\begin{equation}
b\le\frac{x^2}{2}+\frac{x}{2}\sqrt{9x^2+16x+8}.
\end{equation}
Since $\kfour\le3/4$, it is enough to have $b\le8A$. This is equivalent to $27x^2-2x-1\ge0$, which yields $A_+$. The negative-kurtosis case is analogous: now
\begin{equation}
b\le-\frac{x^2}{2}+\frac{x}{2}\sqrt{9x^2+16x+8},
\end{equation}
and $|\kfour|\le3/10$ makes $b\le20A$ sufficient. This is equivalent to $209x^2-2x-1\ge0$ and yields $A_-$. 
\end{proof}

\begin{figure}[t]
\centering
\includegraphics[width=.75\textwidth]{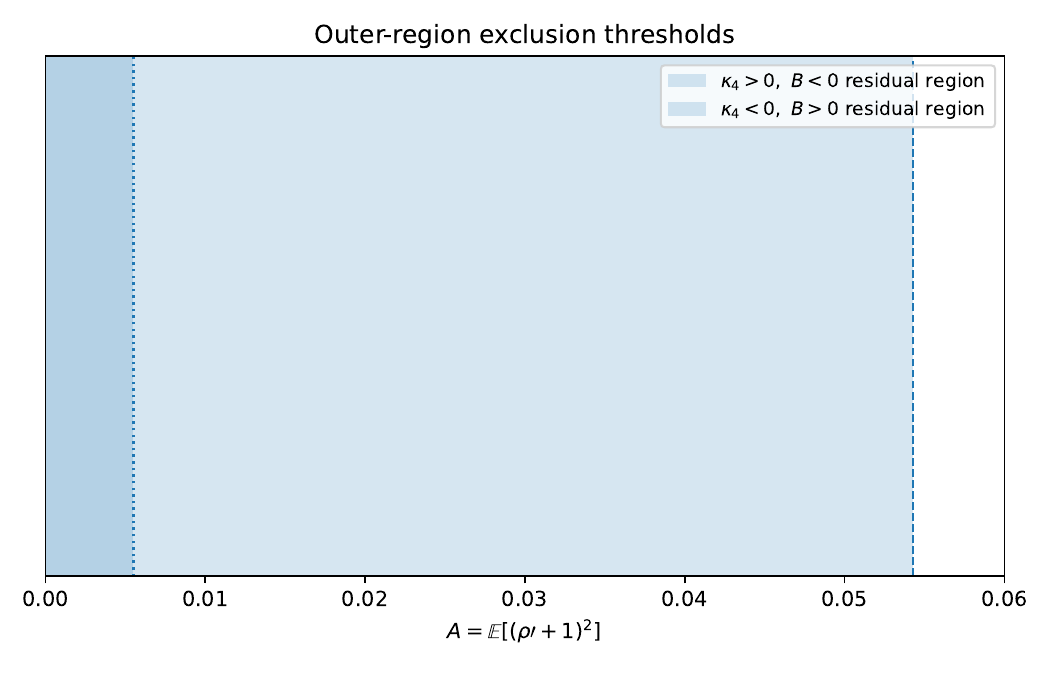}
\caption{The outer-region theorem removes all sign-matched configurations and all sign-mismatched configurations above the indicated information-production thresholds.}
\label{fig:outer}
\end{figure}

\section{Spectral and MMSE reformulations}
\label{sec:spectral}

Let $p=r\phiG$ be an even variance-one density and let $s=(\log r)'$. Define
\begin{equation}
\delta=\E_p[s^2],\qquad A=\E_p[(s')^2].
\label{eq:deltaA}
\end{equation}
The Langevin operator
\begin{equation}
Lf=f''+(\log p)'f',\qquad H=-L,
\end{equation}
is self-adjoint and nonnegative in $L^2(p)$. Since $Hx=x-s$,
\begin{equation}
s=(I-H)x.
\end{equation}
Let $\nu$ be the spectral probability measure of $H$ generated by the coordinate $x$. Since $\E[X^2]=1$,
\begin{equation}
\int\lambda\,\dd\nu(\lambda)=\langle x,Hx\rangle=1.
\end{equation}
Moreover
\begin{align}
\delta&=\int(\lambda-1)^2\,\dd\nu(\lambda),\\
A&=\int\lambda(\lambda-1)^2\,\dd\nu(\lambda),
\end{align}
so
\begin{equation}
\boxed{A-2\delta=\int(\lambda-2)(\lambda-1)^2\,\dd\nu(\lambda).}
\label{eq:spectral}
\end{equation}
This explains why an ordinary spectral-gap estimate alone is insufficient: the sign depends on the coordinate-generated spectral mass relative to both $\lambda=1$ and $\lambda=2$.

At half-Gaussian smoothing, \eqref{eq:posterior-g} gives $s'=2v-1$ and a direct calculation yields
\begin{equation}
\boxed{A-2\delta=4\E[v^2]-1.}
\label{eq:mmse}
\end{equation}
Thus
\begin{equation}
A\ge2\delta
\quad\Longleftrightarrow\quad
\E\!\left[\Var(X\mid X+G)^2\right]\ge\frac14.
\label{eq:mmse2}
\end{equation}
The Gaussian-channel identity $\frac{\dd}{\dd\mathrm{snr}}\mathrm{mmse}=-\E[v^2]$ \cite{GuoShamaiVerdu2005,DytsoPoorShamai2021} gives an equivalent slope formulation.

The relevance of this auxiliary inequality to the endpoint problem is explicit:

\begin{corollary}[Conditional sharpening of theorem~\ref{thm:outer}]
\label{cor:conditional-outer}
In addition to the assumptions of theorem~\ref{thm:outer}, suppose
\begin{equation}
A\ge2\delta.
\label{eq:auxA2d}
\end{equation}
Then the two sign-mismatch thresholds improve to
\begin{equation}
\boxed{A\ge\frac1{26}\quad(\kfour>0,B<0),}
\qquad
\boxed{A\ge\frac1{208}\quad(\kfour<0,B>0).}
\label{eq:improved-thresholds}
\end{equation}
\end{corollary}

\begin{proof}
Equation~\eqref{eq:auxA2d} implies $2\delta\le A$ and hence
\begin{equation}
K=1+A+2\delta\le1+2A.
\end{equation}
Inserting this stronger estimate into \eqref{eq:centeredCS} and repeating the two one-variable bounds in theorem~\ref{thm:outer} gives $A\ge1/26$ and $A\ge1/208$, respectively.
\end{proof}

\section{A rigorous nonlinear neighborhood of the Gaussian}
\label{sec:local}

\begin{theorem}[Nonlinear Gaussian-neighborhood theorem]
\label{thm:local}
Let $p$ be even with mean zero and variance one, write $p=r\phiG$, and assume
\begin{equation}
\norm{r-1}_{\infty}\le\eta<1.
\end{equation}
If
\begin{equation}
2\eta^3-2\eta^2-5\eta+1\ge0,
\label{eq:eta-poly}
\end{equation}
then
\begin{equation}
\boxed{\E_p[(s')^2]\ge2\E_p[s^2].}
\end{equation}
In particular this holds for
\begin{equation}
\boxed{\eta\le\eta_*=0.1884695274\ldots,}
\end{equation}
where $\eta_*$ is the unique root of \eqref{eq:eta-poly} in $(0,1)$.
\end{theorem}

\begin{proof}
Let $M=\E_{\phiG}[s^2]$ and $c=\E_{\phiG}[xs]$. Since $s$ is odd, write $s=cx+u$ with $u$ odd and orthogonal to the first Hermite mode. The remaining odd Gaussian modes have degree at least three, so
\begin{equation}
\E_{\phiG}[(u')^2]\ge3\E_{\phiG}[u^2].
\end{equation}
Therefore
\begin{equation}
\E_{\phiG}[(s')^2]\ge3M-2c^2.
\end{equation}
The exact variance orthogonality \eqref{eq:score-orthogonality} gives
\begin{equation}
c=\E_{\phiG}[(1-r)xs],
\end{equation}
so $|c|\le\eta\sqrt M$. Hence
\begin{equation}
\E_{\phiG}[(s')^2]\ge(3-2\eta^2)M.
\end{equation}
Finally
\begin{equation}
\E_p[(s')^2]\ge(1-\eta)(3-2\eta^2)M,
\qquad
\E_p[s^2]\le(1+\eta)M.
\end{equation}
The factor two follows exactly when \eqref{eq:eta-poly} holds.
\end{proof}

Combining theorem~\ref{thm:local} with corollary~\ref{cor:conditional-outer} makes sections~\ref{sec:spectral}--\ref{sec:local} directly relevant to the endpoint theorem: inside this explicit Gaussian neighborhood, the outer-region constants improve from $(A_+,A_-)$ to $(1/26,1/208)$.

\begin{remark}[Two different notions of ``near Gaussian'']
The hypothesis of theorem~\ref{thm:local} is a global relative-density bound: $(1-\eta)\phiG\le p\le(1+\eta)\phiG$. It does \emph{not} include the fourth-Hermite family of section~\ref{sec:h4}, because for every $a>0$,
\begin{equation}
\left\|\frac{\phiG(1+a\He_4)}{\phiG}-1\right\|_\infty
=a\|\He_4\|_\infty=\infty.
\end{equation}
Thus the local theorem and the Hermite theorem probe complementary perturbative regimes rather than nested neighborhoods.
\end{remark}

\begin{figure}[t]
\centering
\includegraphics[width=.7\textwidth]{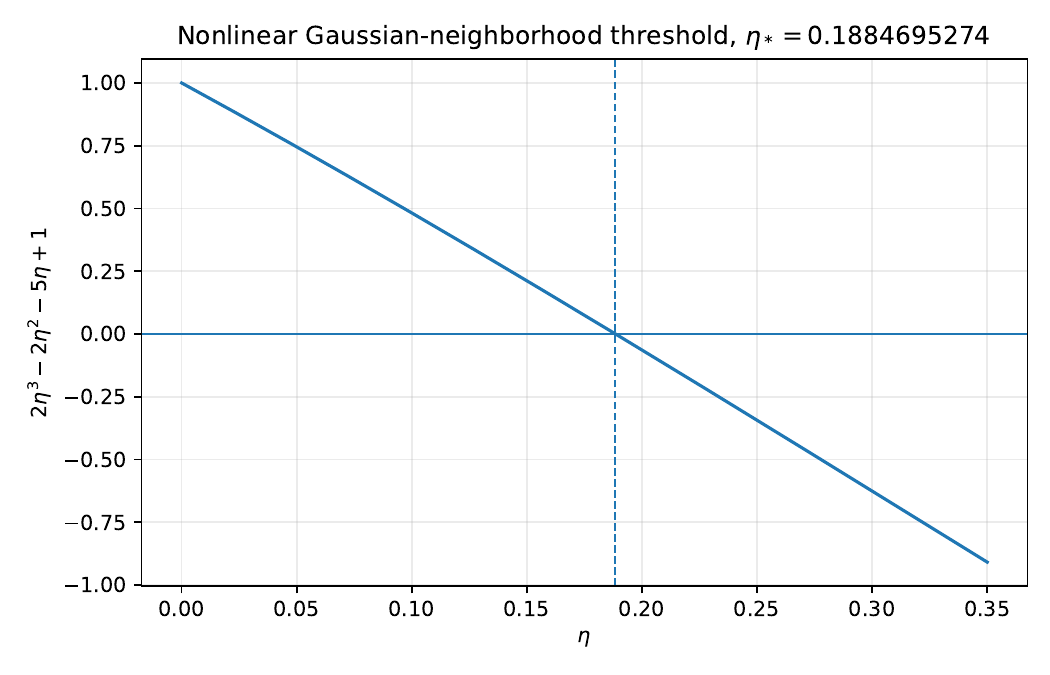}
\caption{The polynomial condition in theorem~\ref{thm:local}. The root is reported consistently as $\eta_*=0.1884695274\ldots$.}
\label{fig:local}
\end{figure}

\section{Noise-induced restoration for the leading symmetric Hermite mode}
\label{sec:h4}

In kinetic theory and near-equilibrium statistical mechanics, deviations from a Maxwellian are often expanded in Hermite or Sonine modes. For an even state with fixed mass, mean and variance, the first possible non-Gaussian mode is the fourth Hermite polynomial. The next theorem is therefore a precise restoration result for the leading symmetric nonequilibrium correction.

Let
\begin{equation}
\He_4(x)=x^4-6x^2+3,
\qquad
f_a(x)=\phiG(x)[1+a\He_4(x)],
\qquad 0<a\le10^{-2}.
\label{eq:h4family}
\end{equation}
Figure~\ref{fig:thermal-h4} shows the damping of this mode under Gaussian thermalization: after applying the $q$-smoothing, the relative density is $1+aq^2\He_4$.

\begin{figure}[t]
\centering
\includegraphics[width=.74\textwidth]{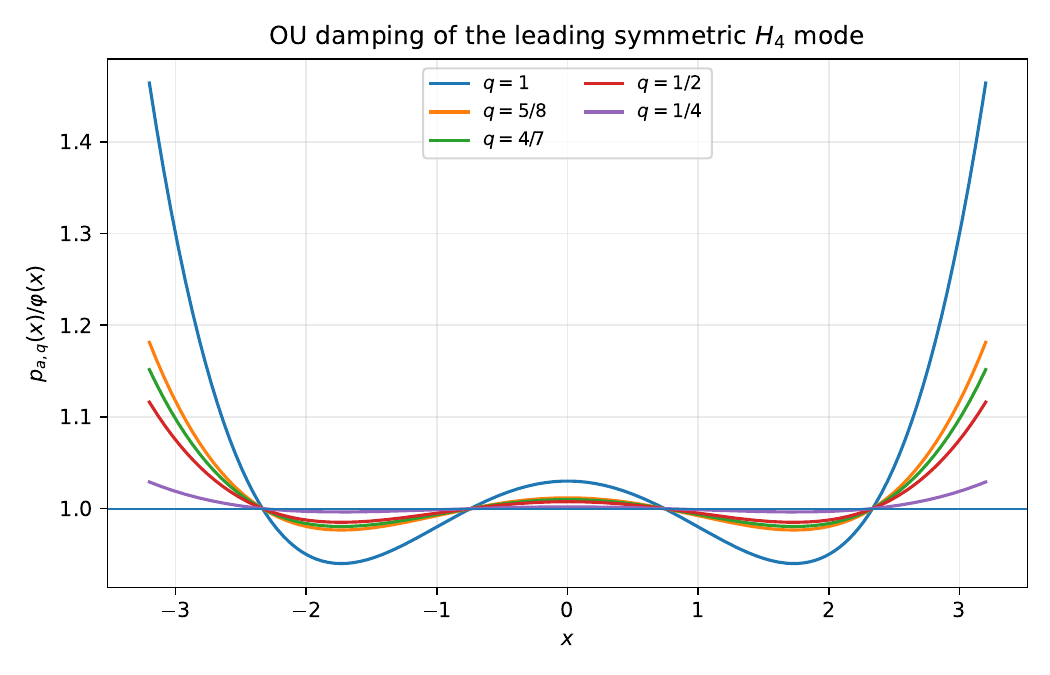}
\caption{Partial thermalization of the leading symmetric nonequilibrium mode. The $\He_4$ amplitude is damped by $q^2$, so Gaussian noise rapidly suppresses the leading kurtosis distortion.}
\label{fig:thermal-h4}
\end{figure}

Gaussian orthogonality gives mass one and variance one. Furthermore
\begin{equation}
P_a^2-P_aP_a''+(P_a')^2
=a^2x^4[(x^2-4)^2+14]+2a(x^2-6)^2+1-54a+45a^2>0,
\end{equation}
so $f_a$ is globally strictly log-concave on the stated amplitude interval.

\begin{theorem}[Uniform concavity for a fourth-Hermite family]
\label{thm:h4}
Let $U,V$ be independent copies of the $q$-smoothed version of \eqref{eq:h4family}. For every
\begin{equation}
0<a\le10^{-2},
\qquad
0<q\le\frac58,
\qquad
0<t<1,
\end{equation}
one has
\begin{equation}
\boxed{F_{a,q}''(t)\le-\frac{(aq^2)^2}{10}<0.}
\label{eq:h4bound}
\end{equation}
In particular, this explicit non-Gaussian family is strictly concave throughout $0<q\le4/7$.
\end{theorem}

\begin{proof}
Put $\epsilon=aq^2$. Since $a\le1/100$ and $q\le5/8$,
\begin{equation}
0<\epsilon\le\frac1{256}.
\end{equation}
Let $z=t(1-t)$ and $S=(1-t)^2+t^2$. The output relative density is exactly
\begin{equation}
R=1+w,
\qquad
w=\epsilon S\He_4+\epsilon^2z^2\He_8,
\label{eq:h4R}
\end{equation}
where
\begin{equation}
\He_8=x^8-28x^6+210x^4-420x^2+105.
\end{equation}
An orthogonal-change-of-variables representation gives
\begin{equation}
R\ge(1-6\epsilon)^2=:m>0.
\end{equation}
The entropy curvature is
\begin{equation}
F''=-\E_{\phiG}\!\left[w_{tt}\log(1+w)+\frac{w_t^2}{1+w}\right].
\label{eq:h4curv}
\end{equation}
The exact quadratic piece is
\begin{equation}
Q_2=-192\epsilon^2(1-3z)-40320\epsilon^4z^2(6-28z)
\le-48\epsilon^2+2520\epsilon^4.
\label{eq:Q2}
\end{equation}
For $1+w\ge m$,
\begin{equation}
|\log(1+w)-w|\le\frac{w^2}{2m},
\qquad
\left|\frac1{1+w}-1\right|\le\frac{|w|}{m}.
\end{equation}
Writing $H=|\He_4|$ and $J_8=|\He_8|$ gives
\begin{align}
|w|&\le\epsilon\left(H+\frac{\epsilon}{16}J_8\right),\\
|w_t|&\le\epsilon\left(2H+\frac{\epsilon}{5}J_8\right),\\
|w_{tt}|&\le\epsilon(4H+2\epsilon J_8).
\end{align}
Consequently
\begin{equation}
\frac{|F''-Q_2|}{\epsilon^2}\le
\frac{\epsilon}{(1-6\epsilon)^2}
\left[
6C_{30}+\frac{23}{10}\epsilon C_{21}
+\frac{713}{3200}\epsilon^2C_{12}
+\frac{41}{6400}\epsilon^3C_{03}
\right].
\label{eq:remainder-bound}
\end{equation}
The four absolute Gaussian moments are evaluated by splitting the real line at the real roots of $\He_4$ or $\He_8$ and integrating each polynomial-Gaussian piece with the standard recurrence for $\int x^n\phiG(x)\dd x$. The resulting values and conservative integer ceilings are shown in table~\ref{tab:h4moments}. In particular we may take
\begin{equation}
(C_{30},C_{21},C_{12},C_{03})=(1800,86000,13700000,4800000000).
\end{equation}
All coefficients in \eqref{eq:remainder-bound} are positive, so the upper bound is increasing in $\epsilon$. At $\epsilon=1/256$ the entire comparison is rational and gives
\begin{equation}
\frac{F''}{\epsilon^2}
\le
-\frac{3703371}{10240000}
=-0.36165732421875
<-\frac1{10}.
\end{equation}
This proves \eqref{eq:h4bound} on the full parameter box.
\end{proof}

\begin{table}[t]
\centering
\begin{tabular}{lrrr}
\toprule
Moment & root-split value & ceiling used & old Cauchy--Schwarz ceiling \\
\midrule
$\E|\He_4|^3$ & $1787.40340119$ & $1800$ & $2973$ \\
$\E[\He_4^2|\He_8|]$ & $85636.3171956$ & $86000$ & $121821$ \\
$\E[|\He_4|\He_8^2]$ & $13600577.4885$ & $13700000$ & $20195129$ \\
$\E|\He_8|^3$ & $4773894164.964$ & $4800000000$ & $11892749013$ \\
\bottomrule
\end{tabular}
\caption{Absolute Gaussian moments entering the exact $H_4$ certificate. The values come from piecewise polynomial-Gaussian integration between the real Hermite roots; the theorem uses only the conservative integer ceilings in the middle column. The last column shows why the original $q\le1/2$ certificate was artificially tight.}
\label{tab:h4moments}
\end{table}

The conservative ceilings above actually keep the $-\epsilon^2/10$ certificate valid up to $q\simeq0.6265$ at $a=10^{-2}$; we state the exact rational range $q\le5/8$ to avoid attaching the theorem to a decimal optimization. The result should also be interpreted carefully: it does not solve the symmetric conjecture on $0<q\le4/7$ for arbitrary log-concave inputs; it proves strict concavity throughout that entire signal range for this explicit family.

This family is not a Gaussian scale mixture. Indeed $\E[\He_8(X)]=0$ but $\E[\He_4(X)]=24a>0$. If $X=SG$ with $S\ge0$ independent of a standard Gaussian $G$, then $\E[\He_8(SG)]=105\E[(S^2-1)^4]$, whose vanishing forces $S^2=1$ almost surely.

\begin{figure}[t]
\centering
\includegraphics[width=.74\textwidth]{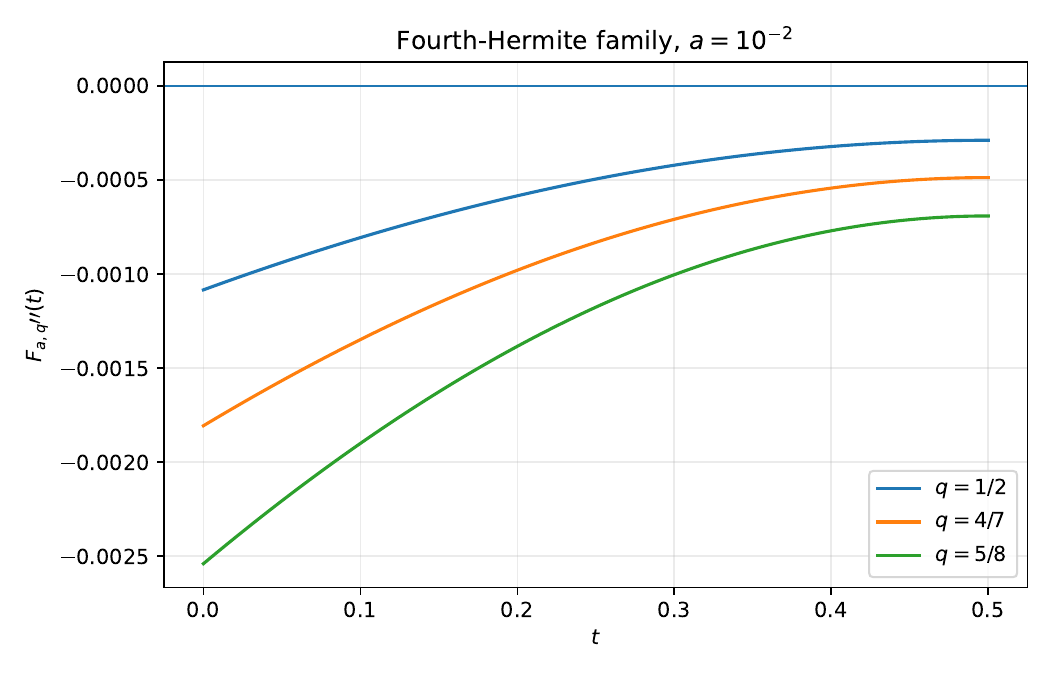}
\caption{Direct quadrature for the fourth-Hermite family at $a=10^{-2}$ for three signal fractions: $q=1/2$, $q=4/7$, and the theorem endpoint $q=5/8$. These curves illustrate the result; the sign theorem itself follows from the analytic certificate above.}
\label{fig:h4}
\end{figure}

\section{Numerical evidence and validation}
\label{sec:numerics}

The ancillary code bundle also contains exploratory searches over finite Hermite expansions and over even convex piecewise-quadratic potentials. These searches are deliberately separated from the theorem statements. The Hermite scan tests positivity and log-concavity on a finite spatial grid and therefore cannot certify global admissibility; the piecewise-convex scan imposes convexity structurally but still evaluates entropy curvature numerically. In the larger endpoint scan no positive curvature was found among $5000$ random symmetric convex-potential inputs at half-Gaussian smoothing. The sign-mismatch example quoted above was stable under quadrature refinement and an independent Fourier implementation.

The numerical layer has three purposes only: validation of analytic formulas, discovery of false shortcuts, and identification of parameter regions where a rigorous theorem might be sharpened. It is not used to infer universal concavity.

\section{Discussion and outlook}

From the statistical-physics perspective, the results quantify when contact with a Gaussian heat bath restores an entropy-stability property that fails in full generality without noise. The strongest all-mixing theorem concerns the leading parity-even nonequilibrium mode around equilibrium and gives an explicit thermalization threshold $q\le5/8$.

The exact endpoint identity \eqref{eq:endpoint-main} turns the symmetric weighted-sum problem into the scalar inequality \eqref{eq:scalar-target}. The weak-regularity lemma is important here: generalized Gaussians with $3/2<p<3$ do not satisfy the naive $C^4$ density expansion at the origin, yet the endpoint quadratic coefficient remains finite and is controlled by the $L^2(f)$ curvature of the potential. Thus the positive generalized-Gaussian result is not an artifact of an unjustified smooth Taylor series.

The half-Gaussian outer theorem shows that any endpoint obstruction must be simultaneously sign-mismatched and close to Gaussian in the information-production variable $A$. Sections~\ref{sec:spectral}--\ref{sec:local} then isolate the auxiliary inequality $A\ge2\delta$, identify it with a spectral moment and an MMSE-slope bound, and make its relevance explicit through corollary~\ref{cor:conditional-outer}. The global $L^\infty$ Gaussian neighborhood and the fourth-Hermite amplitude family are intentionally complementary: neither contains the other.

The strengthened fourth-Hermite theorem is the cleanest all-weight result in the paper. The improved root-split absolute moments eliminate most of the slack in the original certificate and extend the theorem from $q\le1/2$ to the exact rational range $q\le5/8$. In particular, this explicit non-Gaussian, non-Gaussian-mixture family is strictly concave throughout the signal interval $0<q\le4/7$ that is especially relevant to the unresolved symmetric problem.

We do not claim that the full symmetric conjecture is solved. The remaining challenge is to convert the structural posterior, spectral and Fisher-information information into a global inequality without imposing either a relative-density $L^\infty$ bound or a finite-dimensional Hermite ansatz.

\acknowledgments
We used computer algebra and numerical experiments for verification and exploration. Every theorem statement in the paper is supported by an analytic proof or an explicit reproducible certificate; exploratory searches are labeled separately. {\it The use of AI disclosure}: Parts of the numerical cross-checks, the figures and the editing of the text were carried out with the help of an AI assistant (ChatGPT); all results were checked by the author, who takes full responsibility for the content.

\appendix
\section{Weak-regularity endpoint limit for generalized Gaussians}
\label{app:weak}

The point of this appendix is to make clear what is, and is not, being differentiated at the cusp of $|x|^p$. Let
\begin{equation}
V_\varepsilon(x)=a_\varepsilon(x^2+\varepsilon^2)^{p/2},
\qquad p>\frac32,
\end{equation}
with $a_\varepsilon$ fixed by unit variance. Then
\begin{align}
V_\varepsilon'(x)&=a_\varepsilon p\,x(x^2+\varepsilon^2)^{p/2-1},\\
V_\varepsilon''(x)&=a_\varepsilon p(x^2+\varepsilon^2)^{p/2-2}\left[\varepsilon^2+(p-1)x^2\right].
\end{align}
For $|x|\le1$,
\begin{equation}
|V_\varepsilon''(x)|^2\le C_p(x^2+\varepsilon^2)^{p-2}.
\end{equation}
The integral of the right-hand side is uniformly bounded near zero iff $p>3/2$. Similarly $|V_\varepsilon'|^4$ is uniformly integrable, and the exponential tails dominate all polynomial moments. Consequently the three quantities appearing in \eqref{eq:endpoint-V} converge.

For the entropy itself, one may write the $t$-difference quotient in terms of the weak first two coefficients of the convolution/scaling semigroup. The $V''\in L^2(f)$ and $V'\in L^4(f)$ bounds above give a common integrable majorant for the quadratic term. Thus
\begin{equation}
F_p(t)=F_p(0)+F_p'(0^+)t+\frac12F_p''(0^+)t^2+o(t^2),
\end{equation}
and the coefficient is the limit of the smooth coefficients. No $C^3$ statement is needed. For non-even-integer $p$ one should in fact expect nonanalytic higher-order terms from the cusp; this is why theorem~\ref{thm:smooth-endpoint} is not applied directly to $f_p$.

\section{Useful integration-by-parts identities}
For a smooth positive density $f$ with score $\rho=f'/f$ and adequate decay,
\begin{equation}
\E[\rho(X)\psi(X)]=-\E[\psi'(X)].
\end{equation}
Taking $\psi=\rho^3/3$ gives \eqref{eq:score-id}; taking $\psi=x$ gives $\E[X\rho]=-1$. When $f=e^{-V}/Z$, $\rho=-V'$ and $\E[V'']=\E[(V')^2]$.

\section{Gaussian-channel identities at half smoothing}
For $Z=X+G$ and $m(z)=\E[X\mid Z=z]$, Tweedie's identity gives
\begin{equation}
(\log p_Z)'(z)=m(z)-z.
\end{equation}
Differentiating and using $m'(z)=\Var(X\mid Z=z)$ gives
\begin{equation}
(\log p_Z)''(z)=v(z)-1.
\end{equation}
After rescaling $U=Z/\sqrt2$ this becomes \eqref{eq:posterior-g}. Higher conditional derivatives can be organized through posterior cumulants \cite{DytsoPoorShamai2021}.

\section{Absolute-moment bookkeeping for theorem~\ref{thm:h4}}
The four moments in table~\ref{tab:h4moments} are integrals of an absolute polynomial against a Gaussian. The sign can change only at roots of $\He_4$ or $\He_8$. On each sign interval the integrand is an ordinary polynomial times $\phiG$. For
\begin{equation}
I_n(a,b)=\int_a^b x^n\phiG(x)\,\dd x,
\end{equation}
the recurrence
\begin{equation}
I_n(a,b)=a^{n-1}\phiG(a)-b^{n-1}\phiG(b)+(n-1)I_{n-2}(a,b)
\end{equation}
reduces every piece to Gaussian boundary terms and $I_0$. The verifier evaluates the resulting finite expressions at high precision and checks the deliberately rounded integer ceilings. The positive roots used by the script are
\begin{align}
\He_4:&\quad 0.7419637843\ldots,\ 2.3344142183\ldots,\\
\He_8:&\quad 0.5390798114\ldots,\ 1.6365190424\ldots,\ 2.8024858613\ldots,\ 4.1445471861\ldots.
\end{align}
The final theorem-level sign decision itself is rational at $\epsilon=1/256$.


\begin{thebibliography}{99}
\bibitem{BallNayarTkocz2016}
K.~Ball, P.~Nayar and T.~Tkocz,
\emph{A reverse entropy power inequality for log-concave random vectors},
Studia Math. \textbf{235} (2016) 17--30,
arXiv:1509.05926.

\bibitem{Luo2026}
C.~Luo,
\emph{Entropy concavity for log-concave random variables: an asymmetric counterexample},
arXiv:2609.11418 (2026).

\bibitem{EskenazisGavalakis2024}
A.~Eskenazis and L.~Gavalakis,
\emph{On the entropy and information of Gaussian mixtures},
Mathematika \textbf{70} (2024) e12246,
arXiv:2308.15997.

\bibitem{EskenazisNayarTkocz2018Gaussian}
A.~Eskenazis, P.~Nayar and T.~Tkocz,
\emph{Gaussian mixtures: entropy and geometric inequalities},
Ann. Probab. \textbf{46} (2018) 2908--2945,
doi:10.1214/17-AOP1242.

\bibitem{EskenazisNayarTkocz2018Moments}
A.~Eskenazis, P.~Nayar and T.~Tkocz,
\emph{Sharp comparison of moments and the log-concave moment problem},
arXiv:1801.07597 (2018).

\bibitem{ZhangAnantharamGeng2018}
X.~Zhang, V.~Anantharam and Y.~Geng,
\emph{Gaussian optimality for derivatives of differential entropy using linear matrix inequalities},
Entropy \textbf{20} (2018) 182,
doi:10.3390/e20030182.

\bibitem{BallNguyen2012}
K.~Ball and V.H.~Nguyen,
\emph{Entropy jumps for isotropic log-concave random vectors and spectral gap},
Studia Math. \textbf{213} (2012) 81--96,
arXiv:1206.5098.

\bibitem{Bizeul2023}
P.~Bizeul,
\emph{Entropy and Information jump for log-concave vectors},
C. R. Math. \textbf{361} (2023) 487--493,
doi:10.5802/crmath.390.

\bibitem{PayaroPalomar2009}
M.~Payar\'o and D.P.~Palomar,
\emph{Hessian and concavity of mutual information, differential entropy, and entropy power in linear vector Gaussian channels},
arXiv:0903.1945.

\bibitem{AbbeZheng2012}
E.~Abbe and L.~Zheng,
\emph{A coordinate system for Gaussian networks},
arXiv:1002.2271.

\bibitem{DytsoPoorShamai2021}
A.~Dytso, H.V.~Poor and S.~Shamai,
\emph{A general derivative identity for the conditional mean estimator in Gaussian noise and some applications},
arXiv:2104.01883.

\bibitem{BrascampLieb1976}
H.J.~Brascamp and E.H.~Lieb,
\emph{On extensions of the Brunn--Minkowski and Pr\'ekopa--Leindler theorems, including inequalities for log concave functions, and with an application to the diffusion equation},
J. Funct. Anal. \textbf{22} (1976) 366--389.

\bibitem{GuoShamaiVerdu2005}
D.~Guo, S.~Shamai and S.~Verd\'u,
\emph{Mutual information and minimum mean-square error in Gaussian channels},
IEEE Trans. Inf. Theory \textbf{51} (2005) 1261--1282.

\bibitem{SaumardWellner2014}
A.~Saumard and J.A.~Wellner,
\emph{Log-concavity and strong log-concavity: a review},
Stat. Surv. \textbf{8} (2014) 45--114.

\bibitem{Risken1989}
H.~Risken,
\emph{The Fokker--Planck Equation: Methods of Solution and Applications},
2nd ed., Springer, Berlin (1989).

\bibitem{Seifert2012}
U.~Seifert,
\emph{Stochastic thermodynamics, fluctuation theorems and molecular machines},
Rep. Prog. Phys. \textbf{75} (2012) 126001.

\bibitem{Villani2009}
C.~Villani,
\emph{Optimal Transport: Old and New},
Springer, Berlin (2009).

\bibitem{Grad1949}
H.~Grad,
\emph{On the kinetic theory of rarefied gases},
Commun. Pure Appl. Math. \textbf{2} (1949) 331--407.

\bibitem{Cercignani1988}
C.~Cercignani,
\emph{The Boltzmann Equation and Its Applications},
Springer, New York (1988).
\end{thebibliography}
\end{document}